\documentclass[a4paper,12pt]{article}
\usepackage{amssymb}
\usepackage{amsthm}
\usepackage{amsmath}
\usepackage[utf8]{inputenc}
\usepackage{array}
\usepackage[T1]{fontenc}
\usepackage{lmodern}
\usepackage{tikz}
\usepackage{vmargin}
\usepackage{enumitem}
\usepackage{multirow}
\setmarginsrb{2.5cm}{1cm}{2.5cm}{1cm}{1cm}{1cm}{1cm}{1cm}

\DeclareMathOperator{\mad}{mad}
\DeclareMathOperator{\ad}{ad}

\title{Partitioning sparse graphs into an independent set and a forest of bounded degree}

\usepackage{hyperref}
\usepackage[affil-it]{authblk}

\author[a]{François Dross}
\author[a]{Mickael Montassier}
\author[a,b]{Alexandre Pinlou}

\affil[a]{{\small Université de Montpellier, CNRS, LIRMM}} 
\affil[b]{{\small Université Paul-Valéry Montpellier 3, CNRS, LIRMM\medskip}} \affil[ ]{{\small 161 rue Ada, 34095 Montpellier
    Cedex 5, France}} \affil[
]{\href{mailto:francois.dross@lirmm.fr,mickael.montassier@lirmm.fr,alexandre.pinlou@lirmm.fr}{\small{\{francois.dross,mickael.montassier,alexandre.pinlou\}@lirmm.fr}}}

\begin{document}

\newtheorem{theo}{Theorem}
\newtheorem*{theo*}{Theorem}
\newtheorem{cor}[theo]{Corollary}
\newtheorem{lemm}[theo]{Lemma}
\newtheorem{prop}[theo]{Property}
\newtheorem{obs}[theo]{Observation}
\newtheorem{conj}[theo]{Conjecture}
\newtheorem{claim}[theo]{Claim}
\newtheorem{config}[theo]{Configuration}
\newtheorem{quest}[theo]{Question}

\maketitle
\begin{abstract}
An $({\cal I},{\cal F}_d)$-partition of a graph is a partition of the vertices of the graph into two sets $I$ and $F$, such that $I$ is an independent set and $F$ induces a forest of maximum degree at most $d$.
We show that for all $M<3$ and $d \ge \frac{2}{3-M} - 2$, if a graph has maximum average degree less than $M$, then it has an $({\cal I},{\cal F}_d)$-partition. Additionally, we prove that for all $\frac{8}{3} \le M < 3$ and $d \ge \frac{1}{3-M}$, if a graph has maximum average degree less than $M$ then it has an $({\cal I},{\cal F}_d)$-partition.
\end{abstract} 

\section{Introduction}

In this paper, unless we specify otherwise, all the graph considered are simple graphs, without loops or multi-edges.

For $i$ classes of graphs ${\cal G}_1,\ldots, {\cal G}_i$, a
\emph{$({\cal G}_1,\ldots, {\cal G}_i)$-partition} of a graph $G$ is a partition of the vertices of $G$ into $i$ sets $V_1,\ldots,V_i$ such that, for all $1\le j \le i$, the graph $G[V_j]$ induced by $V_j$ belongs to ${\cal G}_j$. 

In the following we will consider the following classes of
graphs:
\begin{itemize}
\item ${\cal F}$ the class of forests,
\item ${\cal F}_d$ the class of forests with maximum degree at most
  $d$,
\item ${\Delta}_d$ the class of graphs with maximum degree at most $d$,
\item ${\cal I}$ the class of empty graphs (i.e. graphs with no edges).
\end{itemize}
For example, an $({\cal I},{\cal F},{\Delta}_2)$-partition of $G$ is a
vertex-partition into three sets $V_1,V_2,V_3$ such that $G[V_1]$ is
an empty graph, $G[V_2]$ is a forest, and $G[V_3]$ is a graph with maximum degree at most $2$.
Note that ${\Delta}_0 = {\cal F}_0 = I$ and $\Delta_1 = {\cal F}_1$.

The \emph{average degree} of a graph $G$ with $n$ vertices and $m$ edges, denoted by $\ad(G)$, is equal to $\frac{2m}{n}$. The \emph{maximum average degree} of a graph $G$, denoted by $\mad(G)$, is the maximum of $\ad(H)$ over all subgraphs $H$ of $G$. The \emph{girth} of a graph $G$ is the length of a smallest cycle in $G$, and infinity if $G$ has no cycle.

Many results on partitions of sparse graphs appear in the literature, where a graph is said to be sparse if it has a low maximum average degree, or if it is planar and has a large girth. The study of partitions of sparse graphs started with the Four Colour Theorem \cite{appel1,appel2}, which states that every planar graph admits an $({\cal I},{\cal I},{\cal I},{\cal I})$-partition. Borodin~\cite{Borodin} proved that every planar graph admits an $({\cal I}, {\cal F}, {\cal F})$-partition, and Borodin
and Glebov~\cite{borodin2001partition} proved that every planar graph
with girth at least $5$ admits an $({\cal I}, {\cal F})$-partition. Poh~\cite{Poh} proved that every planar graph admits an $({\cal F}_2,{\cal F}_2,{\cal F}_2)$-partition.

More recently, Borodin and Kostochka~\cite{borodin2014defective} showed that for all $j\ge 0$ and $k \ge 2j+2$, every graph $G$ with $\mad(G) < 2 \left(2-\frac{k+2}{(j+2)(k+1)} \right)$ admits a $({\Delta}_j,{\Delta}_k)$-partition. In particular, every graph $G$ with $\mad(G) < \frac{8}{3}$ admits an $(I,{\Delta}_2)$-partition, and every graph $G$ with $\mad(G)<\frac{14}{5}$ admits an $(I,{\Delta}_4)$-partition. With Euler's formula, this yields that planar graphs with girth at least $7$ admit $(I,{\Delta}_4)$-partitions, and that planar graphs with girth at least $8$ admit $(I,{\Delta}_2)$-partitions. Borodin and Kostochka~\cite{borodin2011vertex} proved that every graph $G$ with $\mad(G) < \frac{12}{5}$ admits an $(I,{\Delta}_1)$-partition, which implies that that every planar graph with girth at least $12$ admits an $(I,{\Delta}_1)$-partition.
This last result was improved by Kim, Kostochka and Zhu~\cite{kim2014improper}, who proved that every triangle-free graph with maximum average degree at most $\frac{11}{9}$ admits an $(I,{\Delta}_1)$-partition, and thus that every planar graph with girth at least $11$ admits an $(I,{\Delta}_1)$-partition.
In contrast with these results, Borodin, Ivanova, Montassier, Ochem and Raspaud~\cite{borodin2010vertex} proved that for every $d$, there exists a planar graph of girth at least $6$ that admits no $(I,{\Delta}_d)$-partition.

It can be interesting to find partitions of sparse graphs into a forest and a forest of bounded degree, that is $(I,{\cal F}_d)$-partitions. Note that if a graph admits an $(I,{\cal F}_d)$-partition, then it admits an $(I,\Delta_d)$-partition, and that an $(I,{\cal F}_1)$-partition is the same as an $(I,\Delta_1)$-partition. Therefore the previous results imply that:
\begin{itemize}
\item for every $d$, there exists a planar graph of girth at least $6$ that admits no $(I,{\Delta}_d)$-partition;
\item every planar graph with girth at least $11$ admits an $(I,{\cal F}_1)$-partition.
\end{itemize}

Here are the main results of our paper:

\begin{theo} \label{t-main}
Let $M$ be a real number such that $M < 3$. Let $d \ge 0$ be an integer, and let $G$ be a graph with $\mad(G) < M$. If $d \ge \frac{2}{3-M} - 2$, then $G$ admits an $({\cal I},{\cal F}_d)$-partition.
\end{theo}

\begin{theo} \label{t-main2}
Let $M$ be a real number such that $\frac{8}{3} \le M < 3$. Let $d \ge 0$ be an integer, and let $G$ be a graph with $\mad(G) < M$.
If $d \ge \frac{1}{3-M}$, then $G$ admits an $({\cal I},{\cal F}_d)$-partition.
\end{theo}

By a direct application of Euler's formula, every planar graph with girth at least $g$ has maximum average degree less than $\frac{2g}{g-2}$. That yields the following corollary:

\begin{cor} \label{c-main}
Let $G$ be a planar graph with girth at least $g$.

\begin{enumerate}
\hypertarget{3.1}{\item If $g \ge 7$, then $G$ admits an $({\cal I},{\cal F}_5)$-partition. \label{c_1}}

\hypertarget{3.2}{\item If $g \ge 8$, then $G$ admits an $({\cal I},{\cal F}_3)$-partition. \label{c_2}}

\hypertarget{3.3}{\item If $g \ge 10$, then $G$ admits an $({\cal I},{\cal F}_2)$-partition. \label{c_3}}
\end{enumerate}
\end{cor}

Corollaries~\hyperlink{3.2}{3.2} and~\hyperlink{3.3}{3.3} are obtained from Theorem~\ref{t-main2}, whereas Corollary~\hyperlink{3.1}{3.1} is obtained from Theorem~\ref{t-main}.
See Table~\ref{tab1} for an overview of the results on vertex partitions of planar graphs presented above.

\begin{table}
  \centering
  \begin{center}
    \begin{tabular}{|l|l|l|}
      \hline
      Classes & Vertex-partitions & References \\ \hline
      \multirow{3}{*}{Planar graphs} & $({\cal I},{\cal I},{\cal I},{\cal I})$ & The Four
                                                                                 Color Theorem \cite{appel1,appel2}\\
              & $({\cal I},{\cal F},{\cal F})$& Borodin \cite{Borodin} \\
              & $({\cal F}_2,{\cal F}_2,{\cal F}_2)$& Poh \cite{Poh} \\ \hline
      \multirow{1}{*}{Planar graphs with girth 5} & $({\cal I},{\cal F})$& Borodin and Glebov \cite{borodin2001partition} \\ \hline
      \multirow{1}{*}{Planar graphs with girth 6} & no $({\cal I},\Delta_d)$& Borodin et al. \cite{borodin2010vertex} \\ \hline
      \multirow{2}{*}{Planar graphs with girth 7} & $({\cal I},\Delta_4)$& Borodin and Kostochka~\cite{borodin2014defective} \\ 
      & $({\cal I},{\cal F}_5)$ & Present paper \\ \hline
      \multirow{2}{*}{Planar graphs with girth 8} & $({\cal I},\Delta_2)$& Borodin and Kostochka~\cite{borodin2014defective} \\ 
      & $({\cal I},{\cal F}_3)$ & Present paper \\ \hline
      \multirow{1}{*}{Planar graphs with girth 10} 
      & $({\cal I},{\cal F}_2)$ & Present paper \\ \hline
      \multirow{1}{*}{Planar graphs with girth 11} 
      & $({\cal I},\Delta_1)$ & Kim, Kostochka and Zhu~\cite{kim2014improper} \\ 
      \hline
    \end{tabular}
  \end{center}

  \caption{Known results on planar graphs.}
  \label{tab1}
\end{table}
\medskip

\section{Proof of Theorem~\ref{t-main} \label{p-main}}

Let $M < 3$, and let $d$ be an integer such that $d \ge \frac{2}{3-M} - 2$.
Let us call a \emph{good} $d$-partition of a graph $G$ a partition $(I,F)$ of the vertices of $G$ such that $I$ is an independent set of $G$, $G[F]$ is a graph with maximum degree $d$, and every cycle in $G[F]$ goes through a vertex with degree $2$ in $G$. Note that for any graph $G$, if $G$ admits a good $d$-partition, then $G$ admits an $({\cal I},{\cal F}_d)$-partition: while there is a vertex $v$ with degree $2$ in $G$ that is in $F$ and has two neighbours in $F$, move $v$ from $F$ to $I$.
Theorem~\ref{t-main} is implied by the following lemma:
\begin{lemm} \label{l-main}
Every graph $G$ with $\mad(G) < M$ has a good $d$-partition.
\end{lemm}

Our proof uses the discharging method. For the sake of contradiction, assume that Lemma~\ref{l-main} is false. Let $G$ be a counter example to Lemma~\ref{l-main} with minimum order.

For all $k$, a vertex of degree $k$, at least $k$, or at most $k$ in $G$ is a $k$-vertex, a $k^+$-vertex, or a $k^-$-vertex respectively. A $(d+1)^-$-vertex is a \emph{small} vertex, and a $(d+2)^+$-vertex is a \emph{big} vertex. Let $v$ be a vertex of $G$. For all $k$, a neighbour of $v$ of degree $k$, at least $k$, or at most $k$ in $G$ is a $k$-neighbour, a $k^+$-neighbour, or a $k^-$-neighbour of $v$ respectively. A neighbour of $v$ that is a big vertex is a \emph{big neighbour} of $v$, and a neighbour of $v$ that is a small vertex is a \emph{small neighbour} of $v$. We start by proving some lemmas on the structure of $G$. Specifically, we prove that some configurations are reducible, and thus cannot occur in $G$.

\begin{lemm} \label{1v}
There are no $1^-$-vertices in $G$.
\end{lemm}

\begin{proof}
Assume there is a $1^-$-vertex $v$ in $G$. The graph $G-v$ has one fewer vertex than $G$, and thus, by minimality of $G$, admits a good $d$-partition $(I,F)$. If $v$ has no neighbours in $I$, then we can add it to $I$. Otherwise, it has no neighbours in $F$, and we can add it to $F$. In both cases, that leads to a good $d$-partition of $G$, a contradiction.
\end{proof}

\begin{lemm} \label{2v}
Every $2$-vertex has at least one big neighbour.
\end{lemm}

\begin{proof}
Assume $v$ is a $2$-vertex adjacent to two small vertices, $u$ and $w$. The graph $G-v$ has one fewer vertex than $G$, and thus, by minimality of $G$, admits a good $d$-partition $(I,F)$. If $u$ and $w$ are both in $F$, then we can put $v$ in $I$, and if they are both in $I$, then we can put $v$ in $F$. Therefore without loss of generality, we can assume that $u \in I$ and $w \in F$. If $w$ has no neighbours in $I$, then we can put it in $I$, and put $v$ in $F$. Therefore we can assume that $w$ has at least one neighbour in $I$ and thus at most $d-1$ neighbours in $F$ (since $w$ is a small vertex in $G$). Then $v$ has at most one neighbour in $F$, and this neighbour $w$ has at most $d-1$ neighbours in $G[F]$, thus we can add $v$ to $F$.
In every case, this leads to a good $d$-partition of $G$, a contradiction.
\end{proof}

A $2$-vertex is a \emph{leaf} if it is adjacent to a small vertex, and it is a \emph{non-leaf $2$-vertex} otherwise. Note that by Lemma~\ref{2v}, each $2$-vertex has at most one small neighbour. 

\begin{lemm} \label{nobud}
Let $B$ be a set of small $3^+$-vertices such that $G[B]$ is a tree. There exists a $3^+$-vertex $v \notin B$ that is adjacent to a vertex of $B$.
\end{lemm}

\begin{proof}
Assume that the lemma is false, that is every vertex that is not in $B$ but has a neighbour in $B$ is a $2$-vertex. By minimality of $G$, $G - B$ admits a good $d$-partition $(I,F)$. For every vertex $v$ in $B$, successively, we put $v$ in $I$ if $v$ has no neighbours in $I$ and we put it in $F$ otherwise. Note that this way a vertex that we add to $F$ has at most $d$ neighbours that are not in $I$, and we cannot make any cycle in $G[F]$ that does not go through a $2$-vertex, since $G[B]$ is a tree. Thus we have a good $d$-partition of $G$, a contradiction.
\end{proof}

Let $B$ be a (maximal) set of small $3^+$-vertices such that:
\begin{enumerate}[label=(\alph*)]
\item $G[B]$ is a tree, 
\item there is only one edge that links a vertex of $B$ to a $3^+$-vertex $u$ outside of $B$,
\item $u$ is a big vertex. 
\end{enumerate}
We call $B$ a \emph{bud} with father $u$.

Let us build the \emph{light forest} $L$, by the following three steps:
\begin{enumerate}
\item While there are leaves that are not in $L$, do the following. Pick a leaf $v$, and let $u$ be the big neighbour of $v$ (that exists by Lemma~\ref{2v}). Add to $L$ the vertex $v$, the edge $uv$, and the vertex $u$ (if it is not already in $L$). Also set that $u$ is the \emph{father} of $v$ (and $v$ is a \emph{son} of $u$). See Figure~\ref{figL}, left. Note that by doing this, we obtain a star forest with only big vertices and leaves. Also note that the set of the big vertices and the set of the leaves are independent sets in $L$ (but not necessarily in $G$).

\item While there are buds that are not in $L$, do the following. Pick a bud $B$. Let $u$ be the father of $B$, and let $v$ be the vertex of $B$ adjacent to $u$. Add $G[B]$ to $L$, as well as the edge $uv$, and the vertex $u$ (if it is not already in $L$). The vertex $u$ is the father of $v$, and the father/son relationship in $B$ is that of the tree $G[B]$ rooted at $v$. See Figure~\ref{figL}, middle.

\item While, for some $k$, there exists a big $k$-vertex $w \in L$ that has $k-1$ sons in $L$ and whose last neighbour is a $2$-vertex that is not in $L$, do the following. Let $v$ be the $2$-neighbour of $w$ that is not in $L$, and let $u$ be the neighbour of $v$ distinct from $w$. Note that $v$ is a non-leaf $2$-vertex (since it was not added to $L$ in Step 1), therefore $u$ is a big vertex. Add to $L$ the vertex $v$, the edges $uv$ and $vw$, and the vertex $u$ (if it is not already in $L$). We set that $v$ is the father of $w$, and that $u$ is the father of $v$. See Figure~\ref{figL}, right. Note that by doing this, $L$ remains a rooted forest, and that each of the set of the big vertices and the set of the $2$-vertices remains independent in $L$. 
\end{enumerate}

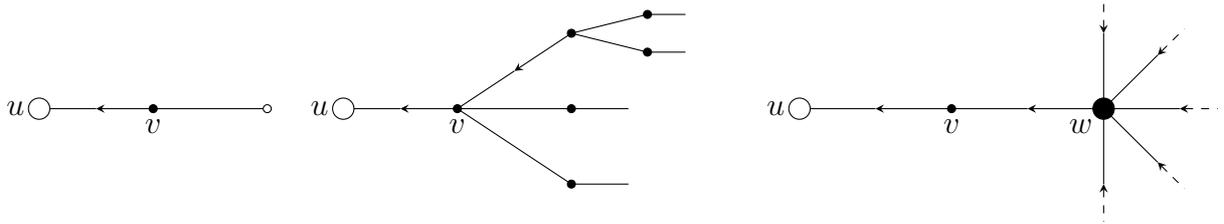
\begin{figure}
\begin{center}
\begin{tikzpicture}
\coordinate (v) at (-10.5,0);
\coordinate (u) at (-12,0);
\coordinate (v') at (-9,0);

\draw (v) node [below] {$v$} ; 
\draw (-12.05,0) node [left] {$u$} ; 

\draw (u) -- (-11.25,0);
\draw [>=stealth,<-] (-11.25,0) -- (v);
\draw (v) -- (v');

\draw [fill=black](v) circle (1.5pt) ;
\draw [fill=white](v') circle (1.5pt) ;
\draw [fill=white](u) circle (4pt) ;

\coordinate (u) at (-8,0);
\coordinate (v) at (-6.5,0);
\coordinate (v1) at (-5,1);
\coordinate (v2) at (-5,0);
\coordinate (v3) at (-5,-1);
\coordinate (v11) at (-4,1.25);
\coordinate (v12) at (-4,0.75);

\draw (v) node [below] {$v$} ; 
\draw (-8.05,0) node [left] {$u$} ; 

\draw (u) -- (-7.25,0);
\draw [>=stealth,<-] (-7.25,0) -- (v);
\draw (v) -- (-5.75,0.5);
\draw [>=stealth,<-] (-5.75,0.5) -- (v1);
\draw (v) -- (v2);
\draw (v2) -- (-4.25,0);
\draw (v) -- (v3);
\draw (v3) -- (-4.25,-1);
\draw (v1) -- (v11);
\draw (v11) -- (-3.5,1.25);
\draw (v1) -- (v12);
\draw (v12) -- (-3.5,0.75);

\draw [fill=black](v) circle (1.5pt) ;
\draw [fill=black](v1) circle (1.5pt) ;
\draw [fill=black](v2) circle (1.5pt) ;
\draw [fill=black](v3) circle (1.5pt) ;
\draw [fill=black](v11) circle (1.5pt) ;
\draw [fill=black](v12) circle (1.5pt) ;
\draw [fill=white](u) circle (4pt) ;

\coordinate (v1) at (0,0);
\coordinate (u1) at (-2,0);
\coordinate (w1) at (2,0); 

\draw (v1) node [below] {$v$} ;
\draw (w1) node [below left] {$w$} ;
\draw (-2.05,0) node [left] {$u$} ;

\draw (u1) -- (-1,0);
\draw [>=stealth,<-] (-1,0) -- (v1);
\draw (v1) -- (1,0);
\draw [>=stealth,<-] (1,0) -- (w1);
\draw (w1) -- (3,0);
\draw [>=stealth,<-] [dashed] (3,0) -- (3.5,0);
\draw (w1) -- (2.71,0.71);
\draw [>=stealth,<-] [dashed] (2.71,0.71) -- (3.06,1.06);
\draw (w1) -- (2,1);
\draw [>=stealth,<-] [dashed] (2,1) -- (2,1.5);
\draw (w1) -- (2.71,-0.71);
\draw [>=stealth,<-] [dashed] (2.71,-0.71) -- (3.06,-1.06);
\draw [<-] (w1) -- (2,-1);
\draw [>=stealth,<-] [dashed] (2,-1) -- (2,-1.5);

\draw [fill=black](v1) circle (1.5pt) ;
\draw [fill=white](u1) circle (4pt) ;
\draw [fill=black](w1) circle (4pt) ;

\end{tikzpicture}
\end{center}
\caption{The construction of the light forest $L$. The big vertices are represented with big circles, and the small vertices with small circles. The filled circles represent vertices whose incident edges are all represented. The dashed lines are the continuation of the light forest. The arrows point from son to father in $L$.\label{figL}}
\end{figure}

As noticed previously, $L$ is a rooted forest. We say that a vertex $v$ is a \emph{descendant} of a vertex $u \ne v$ in $L$ if there are vertices $v_0 = v$, $v_1$, ..., $v_k = u$ in $L$, such that for $i \in \{0,1,...,k-1\}$, $v_{i+1}$
is the father of $v_i$ in $L$.
A vertex $v$ in $L$ is incident to an edge that is not in $L$ only if either $v$ is a big vertex and the root of its connected component in $L$, or $v$ is a leaf, or $v$ is in a bud. The \emph{pending vertices} of $L$ are the vertices that are not in $L$ but are adjacent to a leaf of $L$. Note that the pending vertices are small (by construction).

Let $B$ be a bud with father $u$. Let $S \subseteq V(G) \setminus
(B\cup\{u\})$
and let $(I,F)$ be a good $d$-partition of $S\cup \{u\}$ such that $u$
either is in $I$ or has at most $d-1$ neighbours in $F$. We show that
we can extend the good $d$-partition to $S \cup \{u\} \cup B$. We proceed as
follows: for every vertex $v\in B$, we add $v$ to $I$ if it has no
neighbours in $I$ or to $F$ otherwise. The vertices in $I$ clearly
form an independent set. Moreover, $G[F]$ has maximum degree at most
$d$ and every cycle of $G[F]$ goes through a $2$-vertex by
construction of a bud. This leads to a good $d$-partition of $S\cup B \cup
\{u\}$. We call that
process \emph{colouring the bud $B$}.

Let $v$ be a $2$-vertex of $L$, $u$ its father and $D_v$ the set of the
descendants of
$v$. Let $(I,F)$ be a good $d$-partition of $S \subseteq V(G) \setminus (D_v
\cup \{u,v\})$. We show that we can extend the good $d$-partition to $S \cup
D_v\cup
\{v\}$. We proceed as follows:
\begin{enumerate}
\item[Step $1$.] We add every big vertex of $D_v$ to $I$. Indeed, big vertices
  form an independent set in $L$ and have no neighbours in $S$ by
construction.
\item[Step $2$.]\label{item:pending} Every pending vertex $w\in S$ that has no
neighbours in $I$ is
  added to $I$.
\item[Step $3$.] We add every $2$-vertex of $D_v$ and $v$ to $F$.  Indeed,
  $2$-vertices of $D_v$ form a stable set in $L$. Moreover,
  Step~\ref{item:pending} ensures that the maximum degree of $G[F]$ is
  at most $d$.
\item[Step $4$.] Finally, we colour every bud. Indeed, the father of every bud of
  $D_v$ is in $I$.
\end{enumerate}
This leads to a good $d$-partition of $S\cup D_v \cup \{v\}$. We call that
process \emph{descending $v$}.

\begin{lemm} \label{kroot}
For all $k$, there are no big $k$-vertices in $G$ that are in $L$ and have $k$ sons in $L$.
\end{lemm}

\begin{proof}
Let $u$ be a big $k$-vertex that has $k$ sons in $L$. Note that this implies that $u$ is the root of its connected component in $L$. Let $C$ be the connected component of $u$ in $L$. Let $H = G-V(C)$. The graph $H$ has fewer vertices than $G$ and thus, by minimality of $G$, $H$ admits a good $d$-partition $(I,F)$. Let $N$ be the set of the $2$-neighbours of $u$. We descend every vertex of $N$. Note that this implies that the son of every vertex of $N$ is put in $I$, therefore, up to recolouring the small neighbours of the leaves adjacent to $u$, we can add every vertex of $N$ to $F$ and $u$ to $I$. Then we colour every bud of father $u$. This leads to a good $d$-partition of $G$, a contradiction.
\end{proof}

\subsection*{Discharging procedure}
Let $\epsilon = 3 - M$. 
Recall that $d \ge \frac{2}{3-M} - 2 = \frac{2}{\epsilon}-2$, therefore
 $\epsilon \ge \frac{2}{d+2} > 0$. 
We start by assigning to each $k$-vertex a charge equal to $k - M = k - 3 + \epsilon$. Note that since $M$ is bigger than the average degree of $G$, the sum of the charges of the vertices is negative. The initial charge of each $3^+$-vertex is at least $\epsilon$, and thus is positive.

For every big vertex $v$, $v$ gives charge $1 - \epsilon$ to each of its $2$-neighbours that are its sons in $L$, does not give anything to its father in $L$ (if it has one), and gives $\frac{1 - \epsilon}{2}$ to its other $2$-neighbours.

\begin{lemm}\label{nneg}
Every vertex has non-negative charge at the end of the procedure.
\end{lemm} 

\begin{proof}
The small $3^+$-vertices start with a non-negative charge, and do not give or receive charge throughout the procedure, thus they have non-negative charge at the end of the procedure. 

Every $2$-vertex is either in $L$, in which case it receives $1 - \epsilon$ from its father in $L$, or is not in $L$ and is a non-leaf $2$-vertex, in which case it receives $\frac{1 - \epsilon}{2}$ from each of its neighbours. As $2$-vertices have charge $\epsilon - 1$ in the beginning, and as they receive $1 - \epsilon$, they have charge $0$ at the end of the procedure.

Let $v$ be a big $k$-vertex. By Lemma~\ref{kroot}, $v$ has at most $k-1$ sons in $L$. Moreover, by construction of $L$, if $v$ has  $k-1$ sons in $L$, then either its last neighbour is a $3^+$-vertex, or its last neighbour is its father in $L$ (and in both cases $v$ does not give charge to this vertex).
Therefore $v$ gives charge amounting to at most $(k-1)(1 - \epsilon)$. Since its initial charge is $k - 3 + \epsilon$, in the end it has at least $k - 3 + \epsilon - (k-1)(1 - \epsilon)= k\epsilon - 2$. Since every big vertex has degree at least $d + 2 \ge \frac{2}{\epsilon}$, the final charge of each big vertex is at least $\frac{2}{\epsilon}\epsilon - 2 = 2 - 2 = 0$.

\end{proof}

By Lemma~\ref{nneg}, every vertex has non-negative charge at the end of the procedure, thus the sum of the charges at the end of the procedure is non-negative. Since no charge was created nor removed, this is a contradiction with the fact that the initial sum of the charges is negative. That ends the proof of Lemma~\ref{l-main}.

\section{Proof of Theorem~\ref{t-main2}}

This proof is similar to the proof of Theorem~\ref{t-main} above.
Let $\frac{8}{3} \le M < 3$, and let $d$ be an integer such that $d \ge \frac{1}{3-M}$.
We define good $d$-partitions as in Section~\ref{p-main}.
Theorem~\ref{t-main2} is implied by the following lemma:
\begin{lemm} \label{l-main2}
Every graph $G$ with $\mad(G) < M$ has a good $d$-partition.
\end{lemm}

For the sake of contradiction, assume that Lemma~\ref{l-main2} is false.  Let $G$ be a counter example to Lemma~\ref{l-main2} with minimum order.

We take the same definitions as before. Lemmas \ref{1v}--\ref{kroot} are still true in this setting.
Frank and Gy{\'a}rf{\'a}s~\cite{gyarfas} prove the following theorem:

\begin{theo}[Frank and Gy{\'a}rf{\'a}s~\cite{gyarfas}] \label{gya}
Let $H = (V,E)$ be a graph, and let $\omega : V \rightarrow \mathbb{N}$.
There exists an orientation such that $\forall v \in V, d^+(v) \ge \omega(v)$ if and only if for all $X \subset V$, $\omega(X) \le |\{uv \in E, u \in X\}|$.
\end{theo}

Given $H = (V,E)$ and $\omega : V \rightarrow \mathbb{N}$, a \emph{good $\omega$-orientation} of $H$ is an orientation of $H$ such that $\forall v\in V, d^+(v) \ge \omega(v)$.
We prove some additional lemmas. 

\begin{lemm} \label{omega}
Let $H = (V,E)$ be a graph on $n \ge 1$ vertices and $m$ edges. Let $\omega : V \rightarrow \mathbb{N}$ such that $\omega(V) \le m$.
There exists a subgraph $S$ of $H$ with at least one vertex such that $S$ admits a good $\omega$-orientation.
\end{lemm}

\begin{proof}
For a graph $I$ and a set $X \subset V(I)$, let $e_I(X) = |\{uv \in E(I), u \in X\}|$. If $\omega(X) \le e_I(X)$, we say that $X$ is \emph{good} in $I$.

If every subset of $V$ is good in $H$, then by Theorem~\ref{gya}, we have a good $\omega$-orientation of $H$. Therefore we may assume that there is a subset of $V$ that is not good in $H$. Let $X$ be a maximum subset of $V$ that is not good in $H$. Let $Y = V - X$, and let $H' = H[Y]$. Note that $V$ is good in $H$, since $\omega(V) \le m$, so $Y \ne \emptyset$.

If every subset of $Y$ is good in $H'$, then by Theorem~\ref{gya}, we have a good $\omega$-orientation of $H'$. Therefore there is a $Z \subset Y$ such that $Z$ is not good in $H'$, i.e. $\omega(Z) > e_{H'}(Z)$. As $X$ is not good in $H$, we also have $\omega(X) > e_H(X)$. Therefore we have $\omega(X \cup Z) = \omega(X) + \omega(Z) > e_H(X) + e_{H'}(Z) = |\{uv \in E(H), u \in X\}| + |\{uv \in E(H'), u \in Z\}| = |\{uv \in E(H), u \in (X \cup Z)\}| = e_H(X \cup Z)$. Therefore $X \cup Z$ is good in $H$, which contradicts the maximality of $X$.

\end{proof}

We recall that $L$ is the light forest of $G$.
\begin{lemm} \label{graphH1}
Let $U$ be a non-empty subset of $V(L)$ with no small $3^+$-vertices. Let
$H=G[U]$ (i.e. the subgraph of $G$ induced by the 2-vertices and the
big vertices of $U\subseteq L$). Suppose:
\begin{enumerate}
\item There is an orientation of the edges of $H$ such that every
  2-vertex in $H$ has at least one out-going edge, and for all $i\ge
  1$, every big $(d+i+1)$-vertex in $G$  has at least $i$ out-going
  edges.
\item There are no $1^-$-vertices in $H$.
\end{enumerate}
Then $H$ contains an edge that is not in $L$ and that is incident to a
big vertex.
\end{lemm}

The graph $H$ of Lemma \ref{graphH1} is as follows: it is composed
by subtrees of $L$ plus some additional edges (that do not belong to
$L$). Such edges are edges between leaves, and maybe
edges between roots of trees of $L$. The aim of Lemma \ref{graphH1}
is to prove the existence of such latter edges.

The orientation in Lemma \ref{graphH1} does not correspond to the
orientation defined by the father/son relation. This orientation will
allow us to extend a partial partition $(I,F)$: consider a big
$(d+i+1)$-vertex $v$ being in $F$. Vertex $v$ must have at least $i+1$
neighbours in $I$. The orientation will point towards $i$ sons of $v$ that
will be added to $I$. Moreover we will see that $v$ will have one
extra neighbour in $I$: either its father in $L$, or a neighbour outside $L$.

\bigskip
\begin{proof}
Assume the lemma is false: every edge of $H$ that is not in $L$ is
between two 2-vertices. Let $R_0$ be the set of the vertices of $H$ that
are not the descendants in $L$ of a vertex of $H$. In particular, $R_0$
contains the roots of $L$ that are in $H$, plus big vertices that have no ancestor in $U$. Note that $R_0$ contains only big
vertices ; otherwise, $H$ would contain $1$-vertices. Moreover, $H - R_0$ has at least one vertex, otherwise $U$ would contain only big vertices, there would be an edge between two big vertices, and this edge could not be in $L$. Let $S$ be the set
of the vertices that are not in $H$, but are descendants of vertices
of $H$.

By minimality of $G$, the graph $G - (V(H-R_0) \cup S)$ admits a good
$d$-partition $(I,F)$. While there is a vertex $v\in R_0$ that is in $F$ and
has no neighbours in $I$, we put $v$ in $I$. Now we can assume that
every vertex in $R_0\cap F$ has a neighbour in $I$. Let $R=R_0$ (in
the following we describe a procedure that modifies $R$ but we need to
refer to vertices of $R_0$). While there is
a vertex in $R\setminus (F\cup I)$, do the following:

\begin{itemize}
\item Suppose $u$ is in $I$. We descend every 2-vertex with father $u$
  (by the procedure every 2-vertex is added to $F$) and colour every
  bud with father $u$. This leads to a good $d$-partition of $u$ and all
  its descendants.
\item Suppose $u$ is in $F$. For every 2-vertex $v$ in $H$ with father
  $u$ such that the edge $uv$ is oriented from $u$ to $v$, we
  first add $v$ to $I$, and then add the son of $v$ to $F$ and
  $R$. By hypothesis, if $u$ is a $(d+i+1)$-vertex, then it has at
  least $i$ outgoing edges. These edges lead to sons of $u$:
  \begin{itemize}
  \item either $u\in R_0$, thus $u$ has no ancestors in $H$ by
    construction and all its neighbours in $H$ are its sons ; moreover
    it has a neighbour outside $H$ that is in $I$ (by construction).
  \item or, $u \in R \setminus R_0$, this means that $u$ was added to
    $R$ during the procedure, this implies that his father, say $w$, is
    a 2-vertex added to $I$ and the edge $wv$ is oriented from $w$ to
    $v$ (as every 2-vertex has an out-going edge by hypothesis). It
    follows that all the out-going neighbours of $u$ are sons of $u$.
  \end{itemize}
  It follows that $u$ has at least $i+1$ neighbours in $I$, and so all
  other neighbours can be added in $F$ without violating the degree
  condition on $F$. Now we
  descend every 2-vertex $v\notin H$ with father $u$, and every
  $2$-vertex $v\in H$ with father $u$ such that the edge $uv$ is
  oriented from $v$ to $u$, and colour every bud with father $u$. The
  only problem that could occur is when two adjacent leaves $\ell$ and
  $\ell'$ are added to $I$: in that case, since $\ell$ and $\ell'$
  were added to $I$, the edge that links $\ell$ (resp. $\ell'$) to its
  father is towards $\ell$ (resp. $\ell'$); it follows that one of
  $\ell,\ell'$ has no out-going edges, contradicting the hypothesis.
\end{itemize}
In all cases, that leads to a good $d$-partition of $G$, a contradiction.
\end{proof}

\begin{lemm} \label{graphH2}
Let $U$ be a non-empty subset of $V(L)$ with no small $3^+$-vertices. Let
$H=G[U]$ (i.e. the subgraph of $G$ induced by the 2-vertices and the
big vertices of $U\subseteq L$). Suppose that $H$ has no edge linking
two roots of two connected components of $L$. Let us denote by
$n_2^G(H)$ the number of vertices of $H$ that are 2-vertices in
$G$. Then,
$$|E(H)| < n_2^G(H) + \sum_{big\ v\in H} \left( d_G(v) -d -1 \right).$$
\end{lemm}

\begin{proof}
  By contradiction, suppose there exists such a subgraph $H$ with
  $|E(H)| \ge n_2^G(H) + \sum_{big\ v\in H} \left( d_G(v) -d -1
  \right)$. Let us define a weight function $\omega:V(H)\to
  \mathbb{N}$ such that, for every 2-vertex $u$, $\omega(u)=1$ and,
  for every $(d+i+1)$-vertex $v$, $\omega(v)=i$. By hypothesis,
  $|E(H)|\ge \sum_{v\in V(H)} \omega(v)$. By Lemma \ref{omega}, $H$ contains a
  subgraph $S$ that has a good $\omega$-orientation.

  Suppose that $S$ has a vertex of degree 1, say $v$, and let $u$ be
  the neighbour of $v$ in $S$. As $\omega(v)\ge 1$, the only edge
  incident to $v$ goes from $v$ to $u$. It follows that, for all
  $w\neq v$, $w$ has the same number of outgoing edges in $S$ and in
  $S-\{v\}$. Hence $S-\{v\}$ is a subgraph of $H$ with at least one
  vertex (it contains $u$) and has a good $\omega$-orientation. By
  successively removing vertices of degree 1 from $S$, we can assume
  that $S$ is a subgraph of $H$ that has at least one vertex and that
  admits a good $\omega$-orientation.

  By Lemma \ref{graphH1}, $S$ has an edge $e$ that is not in $L$ and is incident
  to a big vertex. As no leaf of $L$ is adjacent to a big vertex
  except its father, edge $e$ has to link the roots of two connected
  components of $L$, contradicting the hypothesis.
\end{proof}

Let $\widehat L$ be the graph induced by $V(L)$, where we remove every edge that links the roots of two connected components of $L$ and every bud. An \emph{internal} $2$-vertex is a $2$-vertex in $\widehat L$ that has its two neighbours in $\widehat L$. By applying Lemma~\ref{graphH2} to $\widehat L$, we can bound the number of internal $2$-vertices in $\widehat L$. We obtain the following lemma:

\begin{lemm} \label{intern}
The number of internal $2$-vertices is at most $2\sum_{big v \in H} (d_G(v)-d-1$).
\end{lemm}

\begin{proof}
Let $\widehat L'$ be the graph $\widehat L$ where every $2$-vertex is removed in the following way: if $v$ is an internal $2$-vertex with neighbours $u$ and $w$, then we remove $v$ and add an edge from $u$ to $w$, and we iterate. Note that $\widehat L'$ may have multiple edges and even loops. As for each $2$-vertex that was removed, exactly one edge was removed, the number of edges in $\widehat L'$ is at most $\sum_{big\ v\in H} \left( d_G(v) -d -1 \right)$. By Lemma~\ref{2v}, every edge of $\widehat L'$ corresponds to at most two internal $2$-vertices. Therefore there are at most $2\sum_{big v \in H} (d_G(v)-d-1$) internal $2$-vertices.
\end{proof}

\subsection*{Discharging procedure}
Let $\epsilon = 3 - M$ (recall that $\frac{8}{3} \le M \le 3$). Recall that $d \ge \frac{1}{3-M} = \frac{1}{\epsilon}$, therefore $\epsilon \ge \frac{1}{d} > 0$.

We assign to each $k$-vertex a charge equal to $k - M = k-3 + \epsilon$. Note that since $M$ is bigger than the average degree of $G$, the sum of the charges of the vertices is negative.

Every $3^+$-vertex has a charge of at least $\epsilon > 0$. Therefore every vertex that has a negative charge is a $2$-vertex and has charge $\epsilon - 1$. We will redistribute the weight from the $3^+$-vertices to the $2$-vertices, in order to obtain a non-negative weight on each vertex, by the following three steps:

\begin{enumerate}
\item
Let $S$ be a maximal set of small $3^+$-vertices such that $G[S]$ is connected. Let $S_2$ be a set of $2$-vertices that have exactly one (by Lemma~\ref{2v}) neighbour in $S$. Note that since $\epsilon \le 1$, every $k$-vertex in $S$ has charge at least $(k-2)\epsilon$. The vertices in $S$ give $\epsilon$ to each of the vertices in $S_2$. 

Suppose that the total charge of $S$ becomes negative. This implies that the number of vertices in $S_2$ is more than $\sum_{v \in S}(d(v)-2)$. Therefore there are at most $|S|-1$ edges in $G[S]$, and thus $G[S]$ is a tree. Now by Lemma~\ref{nobud}, there is at least one big vertex outside of $S$ that has a neighbour in $S$. Note that if there are at least two of these vertices, or if one of them has at least two neighbours in $S$, then one can observe that $G[S]$ has at most $|S|-2$ edges, contradicting the connectivity of $G[S]$. Therefore $S$ is a bud. In this case, $S$ ends up with a charge of at least $-\epsilon$. We will, in Step $2$, make sure that every son of a big vertex in $L$ receives at least $1 - 2\epsilon \ge \epsilon$ (since $\epsilon \le \frac{1}{3}$) from its father, and this will ensure that every bud ends up with a non-negative charge.

\item
For every big vertex $v$, $v$ gives $1 - 2\epsilon$ to each of its sons, does not give anything to its father (if it has one), and gives $\frac{1 - \epsilon}{2}$ to its other $2$-neighbours. Additionally, every big $k$-vertex gives $2(k - d - 1)\epsilon$ to a common pot.

\item The common pot gives $\epsilon$ to every internal $2$-vertex.
\end{enumerate}

\begin{lemm}\label{nneg2}
Every vertex has non-negative charge at the end of the procedure.
\end{lemm} 

\begin{proof}
Note that by what precedes every small $3^+$-vertex $v$ has non-negative charge.

Every $2$-vertex that is not in $L$ receives $\frac{1 - \epsilon}{2}$ from each of its neighbours. Every leaf that is not adjacent to a $2$-vertex (i.e. every $2$-vertex of $L$ that is not an internal $2$-vertex) receives $1-2\epsilon$ from its father and $\epsilon$ from its other neighbour (which is a small $3^+$-vertex). Every internal $2$-vertex  receives $1-2\epsilon$ from its father and $\epsilon$ from the common pot. Therefore every $2$-vertex has charge $0$ at the end of the procedure.

Let us prove that every big vertex has non-negative charge at the end of the procedure. Let $v$ be a big $k$-vertex. Let $c(v)$ be the initial charge of $v$, and $c'(v)$ be the final charge of $v$. Suppose by contradiction that $c'(v)<0$. By Lemma~\ref{kroot}, vertex $v$ has at most $k-1$ sons. Moreover, if $v$ has $k-1$ sons, then its last neighbour is either the father of $v$, or a $3^+$-vertex (by construction of $L$).  Recall that $\epsilon \le \frac{1}{3}$, therefore $1 - 2\epsilon \ge \frac{1-2\epsilon}{2}$.
If $v$ has $k-1$ sons, then $v$ gives $1-2\epsilon$ to each of its $k-1$ sons, and $2(k - d - 1)\epsilon$ to the common pot, therefore $c'(v) = c(v) - (1-2\epsilon)(k-1) + 2(k - d - 1)\epsilon$, and thus as $c'(v) < 0$, we have $c(v) < (1-2\epsilon)(k-2) + 1-\epsilon + 2(k - d - 1)\epsilon$. If $v$ has $k-2$ sons, then it gives $1-2\epsilon$ to each of its $k-2$ sons, and may give at most $\frac{1 - \epsilon}{2}$ to its two other neighbours and $2(k - d - 1)\epsilon$ to the common pot, therefore $c(v) < (1-2\epsilon)(k-2) + 1-\epsilon + 2(k - d - 1)\epsilon$. If we decrease the number of sons of $v$ further than $k-2$, we will still have $c(v) < (1-2\epsilon)(k-2) + 1-\epsilon + 2(k - d - 1)\epsilon$.

Thus if $c(v) \ge (1-2\epsilon)(k-2) + 1-\epsilon + 2(k - d - 1)\epsilon$, we get a contradiction. Recall that $c(v)$ is equal to $k-3 + \epsilon$.
Therefore we only need to prove that $k-3 + \epsilon\ge (1-2\epsilon)(k-2) + 1-\epsilon + 2(k - d - 1)\epsilon$, which is equivalent to $d \ge \frac{1}{\epsilon}$.
\end{proof}

Let us prove that the common pot also has non-negative charge at the end of the procedure. It receives charge $\sum_{v \text{ big}}2(d(v) - d - 1)\epsilon$. By Lemma~\ref{intern}, this charge is at least $\epsilon$ times the number of internal $2$-vertices. The common pot gives $\epsilon$ to each internal $2$-vertex, therefore it has non-negative charge at the end of the procedure.

By Lemma~\ref{nneg2}, every vertex has non-negative charge at the end of the procedure, thus the sum of the charges at the end of the procedure is non-negative. Since no charge was created nor removed, and since the common pot also has non-negative charge, this is a contradiction with the fact that the initial sum of the charges is negative. That ends the proof of Theorem~\ref{t-main2}.

\bibliographystyle{plain}
\bibliography{biblio} {}

\end{document}